\documentclass[runningheads]{llncs}
\usepackage[T1]{fontenc}
\usepackage{graphicx}
\usepackage{subfigure}
\usepackage[dvips]{epsfig}
\usepackage[dvips]{graphicx}
\usepackage{layout}
\usepackage{amssymb}
\usepackage{todonotes}
\usepackage{complexity}
\usepackage{mathtools}
\usepackage{booktabs}
\usepackage{cleveref}
\usepackage{arydshln}
\usepackage{xspace}
\usepackage{comment}

\newcommand{\dom}{\mathrm{Dom}}
\newcommand{\cl}{\mathrm{Cl}}

\newcommand{\calP}{\mathcal{P}}

\newcommand{\att}{\mathrm{Att}}
\newcommand{\val}{\mathrm{Val}}
\newcommand{\ia}{\ensuremath{\mathrm{IA}}\xspace}
\newcommand{\ias}{\ensuremath{\mathrm{IAs}}\xspace}
\newcommand{\cia}{\ensuremath{\mathrm{CIA}}\xspace}
\newcommand{\pia}{\ensuremath{\mathrm{PIA}}\xspace}
\newcommand{\cias}{\ensuremath{\mathrm{CIAs}}\xspace}
\newcommand{\pias}{\ensuremath{\mathrm{PIAs}}\xspace}

\newif\iffullversion
\fullversionfalse

\newcommand{\nullsymb}{*}
\def\boto{\mkern1.5mu\bot\mkern2.5mu}
\begin{document}
\title{Independence Under Incomplete Information}
%
%
\author{Miika Hannula\inst{1,2}\orcidID{0000-0002-9637-6664} \and
Minna Hirvonen\inst{2,3}\orcidID{0000-0002-2701-9620}\and
Juha Kontinen\inst{2}\orcidID{0000-0003-0115-5154}\and
Sebastian Link\inst{4}\orcidID{0000-0002-1816-2863}}
%
\authorrunning{M. Hannula et al.}
%
\institute{Institute of Computer Science, University of Tartu, Estonia\\ \email{miika.hannula@ut.ee}\\ \and
Department of Mathematics and Statistics, University of Helsinki, Finland\\ \email{juha.kontinen@helsinki.fi}\\ \and Institut für Theoretische Informatik, Leibniz Universität Hannover, Germany\\ \email{minna.hirvonen@thi.uni-hannover.de}  \and
Department of Computer Science, The University of Auckland, New Zealand\\ \email{s.link@auckland.ac.nz}\\}

%
\maketitle              
\begin{abstract}
  We initiate an investigation how the fundamental concept of independence can be represented effectively in the presence of incomplete information in relational databases. The concepts of possible and certain independence are proposed, and first results regarding the axiomatisability and computational complexity of implication problems associated with these concepts are established. In addition, several results for the data and the combined complexity of model checking are presented. The findings help reduce computational overheads associated with the processing of updates and answering of queries. 

\end{abstract}
\section{Introduction}

The concept of independence is appealing to many fields. In databases, it describes when a relation is the cross product for some of its projections~\cite{DBLP:journals/jcss/Paredaens80}. Indeed, the cross product is one of the most fundamental operations, important to designing and querying databases~\cite{DBLP:books/daglib/0006733,DBLP:books/daglib/0011128}. Informally, a relation satisfies the independence atom (IA) $X\boto Y$ whenever for two tuples there is a third that has values matching with the first on all attributes in $X$ and values matching with the second on all attributes in $Y$~\cite{DBLP:journals/jcss/Paredaens80}. In other words, values on $X$ are independent of values on $Y$. For example, the relation in Table~\ref{tab:example} satisfies the IA $\textit{status}\boto \textit{gender}$. IAs form a decidable fragment~\cite{DBLP:journals/jcss/Paredaens80} of embedded multivalued dependencies whose implication problem is undecidable~\cite{herrmann95}. Similarly, probabilistic independence statements form a decidable fragment~\cite{geiger:1991} for the undecidable class of probabilistic conditional independence statements, fundamental in statistics and distributed computing~\cite{DBLP:journals/tit/Li23}. 

\begin{table}
\centering
\caption{Example relation $r$\label{tab:example}}
\begin{tabular}{ccccc}
	\emph{a(ge)} & \emph{e(ducation)} & \emph{s(tatus)} & \emph{r(ace)} & \emph{g(ender)} \\ \hline
	25 & bachelor & not-in-family & white & male \\
	25 & \nullsymb & in-family & white & male \\
	27 & \nullsymb & not-in-family & white & female \\
	\nullsymb & graduate & in-family & \nullsymb & female \\ 
\end{tabular}
\\~\\~\\
\parbox{.45\linewidth}{
\caption{Possible world $w_1$ of relation $r$ from Table~\ref{tab:example}\label{tab:world1}}
\begin{tabular}{ccccc}
	\ \emph{a} \ & \ \emph{e} \ & \ \emph{s}\ & \ \emph{r}\ & \ \emph{g} \ \\ \hline
	25 & bachelor & not-in-fml & white & m \\
	25 & bachelor & in-fml & white & m \\
	27 & graduate & not-in-fml & white & f \\
	27 & graduate & in-fml & white & f 
\end{tabular}
}
\hfill
\parbox{.45\linewidth}{
\caption{Possible world $w_2$ of relation $r$ from Table~\ref{tab:example}\label{tab:world2}}
\begin{tabular}{ccccc}
    \ \emph{a} \ & \ \emph{e} \ & \ \emph{s}\ & \ \emph{r}\ & \ \emph{g} \ \\ \hline
	25 & bachelor & not-in-fml & white & m \\
	25 & graduate & in-fml & white & m \\
	27 & bachelor & not-in-fml & white & f \\
	25 & graduate & in-fml & black& f 
\end{tabular}}
\end{table}

However, database relations are naturally incomplete, where missing values are denoted by null markers (denoted by $\nullsymb$ in this work). Nulls are common since they are used whenever an actual value is unknown at the time of data acquisition~\cite{DBLP:journals/tods/Codd79}, a value simply does not exist~\cite{DBLP:journals/sigmod/Codd86}, or no information about the value is known~\cite{DBLP:journals/vldb/KohlerLLZ16,DBLP:journals/jacm/Lien82,DBLP:journals/jcss/Zaniolo84}. Indeed, nulls accommodate flexibility within the rigid structure that relational databases enforce. If nulls are disallowed in any column, one may specify these columns as \texttt{NOT NULL} in SQL~\cite{DBLP:conf/bncod/Date82,DBLP:journals/sigmod/Grant08}, the de-facto language for defining and querying data. Similar situations occur in statistical models where zeros are often used to denote missing values~\cite{little2002statistical}. 

In this work we ask the following fundamental question: \emph{How can independence be represented when data is missing?} In fact, nulls have led to a broad and deep study of query answering in the presence of incomplete information~\cite{DBLP:journals/jacm/ImielinskiL84,DBLP:journals/tods/Libkin16}, and no single best solution has been proposed~\cite{DBLP:journals/pvldb/ToussaintGLS22}. A promising direction is the approach where nulls are universally interpreted as values that exist but are currently unknown~\cite{DBLP:journals/tods/Codd79}. This leads naturally to a possible world semantics of database relations, where a possible world is obtained by replacing the occurrence of every null marker by some actual domain value. Query answers are then certain or possible. The former are answers in every possible world, while the latter are answers in some possible world~\cite{DBLP:conf/kr/Libkin14}. In the context of our running example, the universal query \emph{Which levels of education are associated with all statuses?} has no certain but possible answers \texttt{bachelor} and \texttt{graduation}, see world $w_1$ in Table~\ref{tab:world1}. 

Interestingly, the worlds $w_1$ in Table~\ref{tab:world1} and $w_2$ in Table~\ref{tab:world2} demonstrate that the IA $e\boto s$ is possible but not certain. In contrast, since neither column \emph{status} nor \emph{gender} contain nulls, IA $s\boto g$ holds necessarily in every possible world, so it is certain. 
Consider the universal query that asks for each status associated with every gender. Since $r$ satisfies the certain IA $s\boto g$, its certain answers are simply all values in column \emph{status}. Generally, the certain (possible) answers to universal queries become certain (possible) answers to simple selection queries whenever the corresponding certain (possible) IA holds. For example, the possible answers to the universal query that returns levels of education associated with all values of status can be obtained by selecting simply all values of education, given the IA $e\boto s$ holds possibly.  

IAs are important for the most fundamental database operations, which are updates and queries. Firstly, IAs may express important semantic constraints that every database instance ought to comply with. As a consequence, updates can only be considered compliant when the instance resulting from the update satisfies every semantic constraint that has been specified on the database schema. This motivates the study of two computational problems associated with IAs: model checking and implication. While model checking refers directly to validating updates, efficient solutions to the implication problem enable us to minimise the overhead for validation, in other words we can reduce validation to non-redundant sets of IAs. Secondly, we have seen how IAs can lead to significant optimisations of expensive queries. For example, expensive universal queries~\cite{DBLP:journals/jcss/LeindersB07} reduce to simple selection queries whenever the underlying relation satisfies a corresponding IA. While checking independence may be as expensive as evaluating the query itself, the optimisation could already be applied if the IA is implied by the given set of constraints. For example, the possible answers to the universal query that returns all combinations of \textit{status} and \textit{gender} affiliated with all levels of \textit{education} cannot be returned as projection on these combinations since the possible IA $e\boto sg$ is not implied by the two possible IAs $e\boto s$ and $es\boto g$, as shown by the relation $r$ in Table~\ref{tab:example}. Indeed, as we will uncover later, the exchange rule, known from classical IAs over complete relations~\cite{DBLP:journals/jcss/Paredaens80,geiger:1991}, does not hold for possible IAs, but it does hold for combinations of possible and certain IAs. This illustrates the challenge and motivates a rigorous study of the underlying model checking and implication problems for possible and certain IAs, and their combination.

Informally, our contributions can be summarised as follows. (1) We propose the concepts of possible and certain IAs, (2) we establish several results regarding the axiomatisability and computational complexity of the implication problem associated with the individual and combined classes for possible and certain IAs, and (3) we establish several results for the data and the combined complexity of model checking for possible and certain IAs.
In the realistic setting of database instances with missing data, we can assign a possible and certain semantics to the classical concept of independence. Similarly to how database queries then have possible and certain answers, update operations may have possible and certain updates. This distinction comes at the prize of an overhead for computing such answers, checking possible and certain models, and deciding possible and certain implication.


\section{Preliminaries}\label{sec:prelim}

We begin the section by defining the underlying data model that accommodates incomplete information. We then introduce the syntax and semantics of possible and certain IAs, and define the implication problems that we will study in the latter sections of this paper.

\subsection{Relations with Incomplete Information}
The natural numbers $\mathbb{N}$ are taken to start from $1$ in this work. Given a natural number $n$, we write $[n]$ for the set $\{1, \dots ,n\}$.
Let $\att$ and $\val$ be disjoint infinite sets of symbols called \emph{attributes} and \emph{values}.
For attribute sets $X$ and $Y$, we often write $XY$ for their set union $X\cup Y$. 
A  \emph{relation schema} is a finite set $R=\{A_1,\ldots,A_n\}$ of attributes from $\mathfrak{A}$. Each attribute $A$ of a relation schema is associated with a domain $\dom(A)$ which is the set of values that can occur in the column $A$. 
In order to allow the data to contain incomplete information about the values of the attributes, we use a special null symbol $\nullsymb$, 
which represents an unknown attribute value. 
 We always assume that $\nullsymb\in\dom(A)$ and $|\dom(A)\setminus\{\nullsymb\}|\geq 2$. The latter assumption is made because an attribute with a domain that has only one non-null value, cannot contain incomplete information, because the unknown value could only be the one non-null value.


A \emph{tuple} over $R$ is a function $t:R\rightarrow\bigcup_{A\in R}\dom(A)$ with $t(A)\in\dom(A)$ for all $A\in R$. The tuple $t$ is called \emph{complete} if does not contain any nulls; that is, $t(A)\neq \nullsymb$ for all $A \in R$---otherwise it is called \emph{incomplete}.
The tuple $t$ is called a null tuple if $t(A)=\nullsymb$ for all $A\in R$. A non-null tuple is a tuple that is not a null tuple. 
For $X\subseteq R$, let $t(X)$ denote the restriction of the tuple $t$ over $R$ on $X$, and $\dom(X)=\prod_{A\in X}\dom(A)$ the Cartesian product of the domains of attributes in $X$. For example, the first tuple of relation $r$ in Table~\ref{tab:example} is complete but all remaining tuples are incomplete. However, the projections of all tuples onto \emph{status} and \emph{gender} are complete. 

A multiset is a pair $M=(B,m)$ consisting of a set $B$ and a multiplicity function $m\colon B\to\mathbb{N}$. The function $m$ determines for each $b\in B$ how many copies of $b$ the multiset $(B,m)$ contains. 
We sometimes say that $M$ \emph{contains} an element $b$, written $b \in M$, if $b$ is in the domain $B$ of $M$.
If $m$ is a constant function for some $n\in\mathbb{N}$, i.e., $m(i)=n$ for all $i\in\mathbb{N}$, we denote it by $n$. For example, $(B,1)$ corresponds to the set $B$ in the usual sense. 
A multiset $(B,m)$ is finite if the set $B$ is finite, and it is a \emph{included} in a multiset $(A,n)$ if $B \subseteq A$ and $m(a) \leq n(a)$ for all $a \in B$.

A \emph{relation} over $R$ is a finite multiset $r=(r',m)$, where $r'$ is a set of tuples over $R$. The relation $r$ is \emph{complete} if it contains only complete tuples, and otherwise incomplete. 
The \emph{projection} of $r$ on $X\subseteq R$ is defined as $r(X)=(r'(X), m')$, where $r'(X)=\{t(X)\mid t\in r'\}$ and $m'\colon r'(X)\to\mathbb{N}$ is such that $m'(t)=\sum_{s(X)=t} m(s)$. For example, the projection of relation $r$ from Table~\ref{tab:example} onto \emph{status} and \emph{race} consists of ((not-in-family,white),2), ((in-family,white),1) and ((in-family,\nullsymb),1). 
We define relations as multisets instead of the more common choice, sets, because we want to allow relations to have more than one copy of tuples that have null symbols. This is because two unknown values that are marked with the null symbol could have two different values if they were known. Our definition also allows multiple copies of tuples without null symbols, but the number of copies for those rows does not affect the satisfaction of the atoms that we will consider.


\subsection{Possible and Certain IAs and Implication Problems}

We introduce possible and certain variants of the IAs. These variants are based on groundings of incomplete relations, each representing a possible world obtained by replacing all null symbols by actual domain values. 

\begin{definition}[Grounding]
For a tuple $t$ and a relation $r$ over a schema $R$,
\begin{itemize}
    \item[(i)] 
    a \emph{grounding} of $t$ is any complete tuple $t'$ over $R$ obtained from $t$ by replacing its null symbols with non-null values (from their respective domains), and
    \item[(ii)] 
    a \emph{grounding} of $r$ is any relation $r'$ over $R$ that is obtained from $r$ by replacing its null symbols with non-null values (from their respective domains) for each copy of a tuple independently. 
\end{itemize}
 \end{definition}
For example, projections of all the groundings of the relation $r$ from Table~\ref{tab:example} onto \emph{status} and \emph{gender} are the same, while they can be different on \emph{education} and \emph{gender}.
In analogy to query answers~\cite{DBLP:conf/kr/Libkin14}, an    IA is possible (certain) whenever it is satisfied by some (all) grounding(s). 


\begin{definition}[Independence, possible and certain independence]
For a relation schema $R$, and $X,Y\subseteq R$, the expressions $X\boto Y$, $X\boto_p Y$, and $X\boto_c Y$ are called \emph{independence atom} (\ia), \emph{possible independence atom} (\pia), and \emph{certain independence atom} (\cia) over $R$, respectively.
For any of these atoms $\sigma$, we write $r\models \sigma$ to mean that a relation $r$ over $R$ satisfies $\sigma$, which is defined as follows:
\begin{itemize}
    \item[(i)] $r\models X\boto Y$ iff  $r(XY)$ is complete, and for all $t_1,t_2\in r$ there is some $t\in r$ such that $t(X)=t_1(X)$ and $t(Y)=t_2(Y)$,
    \item[(ii)] $r\models X\boto_p Y$ iff there is a grounding $r'$ of $r$ such that $r'\models X\boto Y$,
    \item[(iii)] $r\models X\boto_c Y$ iff every grounding $r'$ of $r$ satisfies ${r'\models X\boto Y}$.
\end{itemize}
\end{definition}
An \ia $X \boto Y$ is called \emph{disjoint} if $X$ and $Y$ do not intersect. \emph{Disjoint} \pias and \cias are defined analogously.
Of course, if a relation satisfies an IA, then the IA is certain; and if an IA is certain, then it is also possible. For example, the relation $r$ from Table~\ref{tab:example} satisfies $s\boto g$, $s\boto_c g$, and $s\boto_p g$. Furthermore, $r$ satisfies neither $e\boto_c s$ nor $r\boto_c r$, but it satisfies $e\boto_p s$ and $r\boto_p r$. 



Understanding the interaction between constraints provides us with means to control them in applications. The Introduction has already indicated that a deep understanding of the implication problem for independence statements has direct applications for the most fundamental tasks in data processing: update and query operations.
For a set of atoms $\Sigma$, we write $r\models \Sigma$ if and only if $r\models\sigma$ for all $\sigma\in\Sigma$. 
Let $\Sigma\cup\{\sigma\}$ be a set of atoms over $R$. We say that a set of atoms $\Sigma$ \emph{logically implies} the atom $\sigma$, written as $\Sigma\models\sigma$, if and only if for all relations $r$ over $R$, $r\models\Sigma$ implies $r\models\sigma$. For our running example, $\{e\boto_c s, es\boto_c g\}$ implies $e\boto_c sg$, $\{e\boto_p s, es\boto_p g\}$ does not imply $e\boto_p sg$, and 
$\{e\boto_c s, es\boto_p g\}$ implies $e\boto_p sg$. 
For two subclasses of \ias $\mathcal{P},\mathcal{Q}$, the \emph{$(\mathcal{P},\mathcal{Q})$-implication problem} is to decide whether $\Sigma\models \sigma$, for any finite set of atoms $\Sigma$ from $\mathcal{P}$ and any atom $\sigma$ from $\mathcal{Q}$. If $\mathcal{P}=\mathcal{Q}$, we refer to the $(\mathcal{P},\mathcal{Q})$-implication problem as the  \emph{implication problem for $\mathcal{P}$}.

\section{Axiomatisations}\label{axiomatisations}

Due to the central role of implication problems, our first goal is to establish axiomatic characterisations for combinations of PIAs and CIAs. Throughout, we will observe interesting differences to what is known from the idealised and well-known special case of having complete information. 

\subsection{Inference Rules}

We define possible and certain variants of inference rules known from the special case~\cite{geiger:1991,DBLP:conf/wollic/KontinenLV13,DBLP:journals/jcss/Paredaens80}, and also recall its major results.
Let $\mathcal{T},\mathcal{S},\mathcal{C},\mathcal{D},\mathcal{E}$ be the inference rules for IAs depicted in Table \ref{independence}. We use the subscript $c$ or $p$ for the rule that is obtained from the corresponding rule for IAs by replacing each $\boto$ symbol by the certain independence symbol $\boto_c$ or possible independence symbol $\boto_p$, respectively. The rules with subscript $p\& c$ are given in \Cref{pc-ind}.
We define the following sets of inference rules. Independence: $\mathfrak{I}=\{\mathcal{T},\mathcal{S},\mathcal{C},\mathcal{D},\mathcal{E}\}$, certain independence: $\mathfrak{I}_c=\{\mathcal{T}_c,\mathcal{S}_c,\mathcal{C}_c,\mathcal{D}_c,\mathcal{E}_c\}$, possible independence: $\mathfrak{I}_p=\{\mathcal{T}_p,\mathcal{S}_p,\mathcal{C}_p,\mathcal{D}_p\}$, and mixed exchange rules: $\mathfrak{I}_{p\& c}=\{\mathcal{E}_{p\& c},\mathcal{E}_{c\& p}\}$.
Note that the set $\mathfrak{I}_p$ is similar to the sets $\mathfrak{I}$ and $\mathfrak{I}_c$, except that the exchange rule is missing. This is a necessary omission, as Example \ref{notequiv} demonstrates that the exchange fails for possible independence. The implication of $X\boto_p Y$ by $X\boto_c Y$ follows from mixed exchange $\mathcal{E}_{c\& p}$ and trivial independence $\mathcal{T}_{p}$.
\begin{table}
\parbox{.45\linewidth}{
\[\fbox{$\begin{array}{c@{\hspace*{.25cm}}c}
\cfrac{}{X\boto\emptyset} & \cfrac{X\boto Y}{Y\boto X}\\
\text{(trivial ind., $\mathcal{T}$)} & \text{(symmetry, $\mathcal{S}$)}\\ \\
\cfrac{X\boto X\qquad Y\boto Z}{XY\boto Z}&\cfrac{X\boto YZ}{X\boto Y}
\\
\text{(constancy, $\mathcal{C}$)} & \text{(decomposition, $\mathcal{D}$)}\\ \\
\cfrac{X\boto Y\quad XY\boto Z}{X\boto YZ}  \\
\text{(exchange, $\mathcal{E}$)} 
\end{array}$}\]
\caption{Rules $\mathfrak{I}$ for independence \label{independence}}
}
\hfill
\parbox{.45\linewidth}{
\[\fbox{$\begin{array}{c@{\hspace*{.25cm}}c}
\cfrac{X\boto_p Y\quad XY\boto_c Z}{X\boto_p YZ} &   \\
\text{(mixed exchange, $\mathcal{E}_{p\& c}$)} & \\ \\
\cfrac{X\boto_c Y\quad XY\boto_p Z}{X\boto_p YZ} & \\
\text{(mixed exchange, $\mathcal{E}_{c\& p}$)} &
\end{array}$}\]
\caption{Rules $\mathfrak{J}_{p\& c}$ for possible and certain independence \label{pc-ind}}
}
\end{table}

Let $\mathfrak{A}$ be a set of axioms. We say that there is an $\mathfrak{A}$-deduction for an atom $\sigma$ from the set of atoms $\Sigma$, written as $\Sigma\vdash_{\mathfrak{A}}\sigma$, if and only if there is a finite sequence $(\tau_1,\dots,\tau_n)$ of atoms such that $\tau_n=\sigma$, and each $\tau_i$, $i\in\{1,\dots,n\}$, is either an element of $\Sigma$ or is obtained by applying a rule from $\mathfrak{A}$ to some atoms of $\Sigma\cup\{\tau_1,\dots,\tau_{i-1}\}$. We sometimes write $\Sigma\vdash\sigma$ instead of $\Sigma\vdash_{\mathfrak{A}}\sigma$, if the set of axioms $\mathfrak{A}$ is clear from the context.

Let $\mathfrak{A}$ be an axiomatisation, and $\mathcal{P}$ and $\mathcal{Q}$ classes of atoms. 
    We say that an axiomatisation $\mathfrak{A}$ is (i) \emph{sound} for the $(\mathcal{P},\mathcal{Q})$-implication problem iff for all sets of atoms $\Sigma$ of the class $\mathcal{P}$ and all atoms $\sigma$ of the class $\mathcal{Q}$, $\Sigma\vdash_{\mathfrak{A}}\sigma$  implies $\Sigma\models\sigma$, and (ii) \emph{complete} for the $(\mathcal{P},\mathcal{Q})$-implication problem iff for all sets of atoms $\Sigma$ of the class $\mathcal{P}$ and all atoms $\sigma$ of the class $\mathcal{Q}$, $\Sigma\models\sigma$ implies  $\Sigma\vdash_{\mathfrak{A}}\sigma$.
We write $\cl_{\mathfrak{A}}(\Sigma)$ for the closure of $\Sigma$ with respect to the axiomatisation $\mathfrak{A}$ consisting of axioms, i.e., $\cl_{\mathfrak{A}}(\Sigma)=\{\sigma\mid\Sigma\vdash_{\mathfrak{A}}\sigma\}$. For example, if $\Sigma$ consists of $e\boto_c s$, $es\boto_p g$ and $r\boto_p r$, then $rsg\boto_p e\in\cl_{\{\mathcal{E}_{c\& p},\mathcal{S}_{p},\mathcal{C}_{p}\}}(\Sigma)$ due to the following $\{\mathcal{E}_{c\& p},\mathcal{S}_{p},\mathcal{C}_{p}\}$-deduction:
\[
\begin{array}{lll}
& e\boto_c s & es\boto_p g \\
& \multicolumn{2}{l}{\overline{\mathcal{E}_{c\& p}:\hspace*{.1cm}e\boto_p sg\hspace*{.1cm}}} \\
r\boto_p r & \multicolumn{2}{l}{\overline{\mathcal{S}_{p}:\hspace*{.2cm}sg\boto_p e\hspace*{.2cm}}} \\
\multicolumn{3}{l}{\overline{\mathcal{C}_{p}:\hspace*{.7cm}rsg\boto_p e\hspace*{.7cm}}}
\end{array}\;.\]
For a CIA or PIA $\sigma$ define $ind(\sigma)$ as the corresponding IA, i.e., $ind(X \boto_c Y)=X \boto Y$ and $ind(X \boto_p Y)=X \boto Y$. We extend the definition of $ind$ to sets of CIAs and PIAs in the obvious way, i.e., $ind(\Sigma)=\{ind(\sigma)\mid \sigma\in\Sigma\}$. 

The following result about the idealised special case of complete relations is well-known from the research literature~\cite{geiger:1991,DBLP:conf/wollic/KontinenLV13,DBLP:journals/jcss/Paredaens80}.

\begin{theorem}\label{thm-complete}
The set $\mathfrak{I}$ forms a sound and complete axiomatisation for the implication problem for IAs.
\end{theorem}
Indeed, a stronger result is known~\cite{Fagin82,geiger:1991,DBLP:conf/wollic/KontinenLV13} that allows us to construct a single complete relation, known as an Armstrong relation, that satisfies an IA if and only if it is implied by a given set of IAs. Hence, the implication problem reduces to a model checking problem on such a relation. This is known to be useful for acquiring those constraints perceived to encode business rules of the underlying application~\cite{Fagin82}. 
Let $\Sigma$ be a set of atoms of class $\mathcal{P}$ over a schema $R$. A relation $r$ over $R$ is called an \emph{Armstrong relation} for a set $\Sigma$ of atoms of class $\mathcal{P}$ if for all atoms $\sigma$ of class $\mathcal{P}$ over $R$, $r\models\sigma$ if and only if $\Sigma\models\sigma$. We say that a class $\mathcal{P}$ of atoms enjoys Armstrong relations, if every set of atoms $\Sigma$ of class $\mathcal{P}$ has an Armstrong relation.

\begin{theorem}\label{thm-armstrong-rel}
    The class of IAs enjoys Armstrong relations.
\end{theorem}

\subsection{Results for Certain and Possible IAs Separately}

As our first major result we establish the completeness of the axiomatisation $\mathfrak{I}_c$ for CIAs. This follows from the completeness of $\mathfrak{I}$ for IAs, by showing that if $\Sigma\not\vdash_{\mathfrak{I}_c} \sigma$, then $ind(\Sigma)\vdash_{\mathfrak{I}}ind(\sigma)$, and the relation that witnesses $ind(\Sigma)\not\models ind(\sigma)$ also witnesses $\Sigma\not\models \sigma$. The proof is included in Appendix \ref{appendix}.

\begin{theorem}\label{thm-completecert}
The set $\mathfrak{I}_c$ forms a sound and complete axiomatisation for the implication problem for CIAs.
\end{theorem}
Indeed, the implication problems for IAs and CIAs are equivalent in the following sense.
\begin{theorem}\label{equiv}
    Let $\Sigma\cup\{\sigma\}$ be a set of CIAs.  Then
    $\Sigma\models\sigma \iff ind(\Sigma)\models ind(\sigma).$
\end{theorem}
It follows that IAs in the idealised special case of having complete information correspond exactly to CIAs in the general case of permitting incomplete information. Intuitively, our finding assures us that known results from an idealised case hold in a more general and realistic context in the principled sense of certainty. This view extends further to Armstrong relations.   

\begin{theorem}\label{armstr}
    The class of CIAs enjoys Armstrong relations.
\end{theorem}
Apart from positioning known results in the broader framework of incomplete information, we also need to consider the case where IAs may possibly hold. Here, things are different. Indeed, Theorem \ref{equiv} does not hold if we consider PIAs instead of  CIAs, because the exchange rule is not sound in the possible case. This is demonstrated by the following example.
\begin{example}\label{notequiv}
There exists a set $\Sigma\cup\{\sigma\}$ of PIAs such that
    $ind(\Sigma)\models ind(\sigma)$, but $\Sigma\not\models\sigma$.
    Let $\Sigma=\{A\boto_p B, AB\boto_p C\}$ and $\sigma= A\boto_p BC$. Then $ind(\Sigma)=\{A\boto B, AB\boto C\}$ and $ind(\sigma)= A\boto BC$, and by the soundness of exchange $\mathcal{E}$, $ind(\Sigma)\models ind(\sigma)$. Consider then the relation $r$ depicted in Table \ref{exchangefig}.
    We have $r\models\Sigma$, because $r$ has groundings $r'$ and $r''$ such that $r' \models A\boto B$ and $r''\models AB\boto C$. But $r\not\models A\boto_p BC$, since there is no way to replace all null symbols in the column $A$ such that $A\boto BC$ would hold. This is because $r$ has two different values for $A$ and four different values for $BC$, making it impossible to fit all of the eight different values for $ABC$ in just four rows.
\end{example}
\begin{table}
\parbox{.45\linewidth}{
    \centering
\begin{tabular}{ c c c }  
 $A$ & $B$ & $C$   \\  
 \hline
 0 & 0 & 0  \\
 $\nullsymb$ & 1 & 0  \\
 $\nullsymb$ & 0 & 1 \\
 1 & 1 & 1 \\
\end{tabular}
 \quad
       \begin{tabular}{ c c c }      
$A$ & $B$ & $C$   \\  
 \hline
 0 & 0 & 0  \\
 0 & 1 & 0  \\
 1 & 0 & 1  \\
 1 & 1 & 1  \\  
 \end{tabular}
 \quad
       \begin{tabular}{ c c c }   
 $A$ & $B$ & $C$   \\  
 \hline
 0 & 0 & 0  \\
 1 & 1 & 0  \\
 0 & 0 & 1  \\
 1 & 1 & 1  \\
 \end{tabular}
 \\~\\
   \caption{The relations $r$, $r'$, and $r''$.}
    \label{exchangefig}
}
\hfill
\parbox{.45\linewidth}{
     \begin{tabular}{c c c c c c c c|c}
       $X_1$ & $X_2$ & $X_3$ & $Y_1$ & $Y_2$ & $Z_1$ & $\dots$ & $Z_n$ & $\#$   \\
       \hline
       0 & 0 & 0 & 0 & 0 & 0 & $\dots$ & 0 & 1   \\
       0 & 0 & 1 & 0 & 1 & 0 & $\dots$ & 0 & 1   \\
       0 & 1 & 0 & 1 & 0 & 0 & $\dots$ & 0 & 1   \\
       0 & 1 & 1 & $\nullsymb$ & $\nullsymb$ & 0 & $\dots$ & 0 & 1   \\
       1 & 0 & 0 & $\nullsymb$ & $\nullsymb$ & 0 & $\dots$ & 0 & 1   \\
       1 & 0 & 1 & $\nullsymb$ & $\nullsymb$ & 0 & $\dots$ & 0 & 1   \\
       1 & 1 & 0 & $\nullsymb$ & $\nullsymb$ & 0 & $\dots$ & 0 & 1   \\
       1 & 1 & 1 & $\nullsymb$ & $\nullsymb$ & 0 & $\dots$ & 0 & 1   \\
       $\nullsymb$ & $\nullsymb$ & $\nullsymb$ & $\nullsymb$ & $\nullsymb$ & 0 & $\dots$ & 0 & 15  \\
    \end{tabular}
    \\~\\
    \caption{The relation $r$ from the proof of Thm. \ref{thm-completepos} in the case $k=3$, $m=2$.
    }
    \label{fig:possibleind}}
\end{table}
The invalidity of the exchange rule for PIAs also means that the completeness proof for the axiomatisation of IAs does not work for PIAs. However, the set $\mathfrak{I}_p$ forms a complete axiomatisation for a restricted version of the implication problem for PIAs.
\begin{theorem}\label{thm-completepos}
The set $\mathfrak{I}_p$ 
forms a sound and complete axiomatisation for the $(\pia,\pia^*)$-implication problem, where $\pia^*$ is the class of \pias $X\boto_p Y$ such that at least one of the following conditions hold:

\centering
    (i) $|X|=1$ or $|Y|=1$, \qquad or \qquad (ii) $||X|-|Y||\leq 1$.
\end{theorem}
\begin{proof}
    The soundness of the axiomatisation is again easy to show by checking that all of the inference rules are sound.
    We show that the axiomatisation is complete. Let $\Sigma$ be a set of PIAs and $X\boto_p Y$ a PIA. Suppose that $\Sigma\not\vdash X\boto_p Y$. We show that $\Sigma\not\models X\boto_p Y$. Note that if $\Sigma\vdash A\boto_c A$, we can use constancy $\mathcal{C}_p$ to obtain $\Sigma\not\vdash X\setminus\{A\}\boto_p Y\setminus\{A\}$, and it suffices to show the claim for this atom. Hence, we may assume that there are no $A\in XY$ such that $\Sigma\vdash A\boto_p A$. 
   The case that $X\boto_p Y$ is non-disjoint is considered in the full version of the proof in Appendix \ref{appendix}.
   The case that $X\boto_p Y$ is disjoint can be handled in two parts. The first case that $|X|-|Y|\leq 1$ and $|Y|\geq 2$. The second case is that $|Y|=1$. Note that by symmetry we may assume that $|X|\geq |Y|$, so it suffices to consider these two cases. 
We consider the first case only. The proof of the second case can be found in the full version of the proof.

    For the first case, let $|X|=k\geq m=|Y|$, and assume that $m\geq 2$. Let $\dom(A)=\{0,1,\nullsymb\}$ for all $A\in R$. Construct then a relation $r$ over $R$ as follows. The relation $r$ consists of $2^{k}(2^m-1)-1$ rows. On the first $2^{k}$ rows, let the values of $X$ be the tuples from $\{0,1\}^{k}$. For the first $2^m-1$ rows of these $2^{k}$ rows, let the values of $Y$ be the tuples from $\{0,1\}^m\setminus \{(1,1,\dots,1)\}$ and for the rest of the $2^{k}-(2^m-1)$ rows, let the values of $Y$ be null symbols. Then add rows with only null symbols such that you obtain $2^{k}(2^m-1)-1$ rows in total. Let all the values for the attributes $A\in R\setminus XY$ be constant 0.
    
We show that $r\models\Sigma$, but $r\not\models X\boto_p Y$. First note that $r(X)$ has $2^{k}$ different complete tuples, and $r(Y)$ has $2^m-1$ different complete tuples. 
    Therefore, it is impossible for $X\boto_p Y$ to hold in the relation $r$, because $r$ has only $2^{k}(2^m-1)-1$ tuples.
     Suppose then that $V\boto_p W \in\Sigma$. We will show that $r\models V\boto_p W$. Since every attribute $A\in R\setminus XY$ is constant, we may assume that $VW\subseteq XY$, and therefore also $V\cap W=\emptyset$. (For the latter, recall that there are no $A\in XY$ such that $\Sigma\vdash A\boto_p A$.) Moreover, we may assume that $VW=XY$, because by decomposition $\mathcal{D}_p$, $r\models V\boto_p W$ implies $r\models V'\boto_p W'$ for every $V'\subseteq V$ and $W'\subseteq W$. 
    We may assume that $|V|\geq|W|$. 
    
    Suppose first that $|V|\geq k+1$. This means that $r(V)$ has at most $2^{k}$ different non-null tuples and $r(W)$ has at most $2^{k+m-|V|}\leq 2^{m-1}$ different non-null tuples. Then in order to $V\boto_p W$ hold, we need to fit at most $2^{k}\cdot 2^{m-1}$ values of $XY$ in the relation. Since we have $2^{k}(2^m-1)-1(>2^{k}\cdot 2^{m-1})$ rows in $r$, and no non-null tuple is repeated in $r$, we have enough room to contain all the needed rows. This can be done by a grounding $r'$ of $r(XY)$ that replaces all the null symbols on the first $2^k$ rows with 0s, and replaces the null symbol rows with tuples from the Cartesian product of the projections of these first $2^k$ rows on $V$ and $W$. Since replacing all the null symbols on the first $2^k$ rows with 0s preserves the amounts of different non-null tuples, the cardinality approximations above hold for $r'$. Thus $r\models V\boto_p W$.

    Suppose then that $|V|=k$. Since $|V|+|W|=|VW|=|XY|=|X|+|Y|=k+m$, we have $|W|=m$. Since $V\neq X$ and $W\neq X$, there are $A\in V$ and $B\in W$ such that $A,B\in Y$. As the relation $r$ does not contain the tuple where every attribute of $Y$ has value 1, it must be that by picking suitable values for the null symbols on the first $2^k$ row, we obtain at most $2^{k}-1$ different values for $V$ and $r(W)$ has at most $2^m-1$ different values for $W$. Since no non-null tuple is repeated in $r$, we have enough room to contain all the needed rows. This can be done by a grounding $r'$ of $r(XY)$ that replaces the null symbols on the first $2^k$ rows such that we obtain at most $2^{k}-1$ and $2^{m}-1$ different values for $V$ and $W$, respectively. The grounding $r'$ replaces the null symbol rows with tuples from the Cartesian product of the projections of these first $2^k$ rows on $V$ and $W$, as before. Thus $r\models V\boto_p W$.
Note that it cannot be that $|V|\leq k-1$. Otherwise, from $|X|-|Y|\leq 1$ and $|X|=k$, it follows that  $k+(k-1)\leq|X|+|Y|=|XY|=|VW|=|V|+|W|\leq k-1+|W|$, i.e., $|W|\geq k$, which is impossible because $|W|\leq|V|\leq k-1$. This finishes the proof in the case that $m\geq 2$. Table \ref{fig:possibleind} depicts the relation $r$ in the case $k=3$, $m=2$, where $X=\{X_1,X_2,X_3\}$, $Y=\{Y_1,Y_2\}$, and $R\setminus XY=\{Z_1,\dots,Z_n\}$.

\end{proof}
 For our running example, the incomplete relation $r$ of Table~\ref{tab:example} shows that $\sigma=e\boto_p sg$ is not implied by the set $\Sigma$ consisting of $e\boto_p s$ and $es\boto_p g$. Indeed, while grounding $w_1$ of $r$ in Table~\ref{tab:world1} satisfies $e\boto s$ and grounding $w_2$ of $r$ in Table~\ref{tab:world2} satisfies $es\boto p$, there cannot be any grounding of $r$ that satisfies both. Theorem~\ref{thm-completepos} shows that there is no $\mathfrak{I}_p$-deduction of $\sigma$ from $\Sigma$.
It is future work to investigate the axiomatisability for other fragments of PIAs.

\subsection{Combining Certain and Possible IAs} 

The most general implication problem combines the classes of CIAs and PIAs. In this context, the following theorem shows that, when restricted to disjoint atoms, we may add PIAs to the set $\Sigma$, and still obtain a complete axiomatisation.
\begin{theorem}\label{thm-completedisjointcertpos_cert}
The set $(\mathfrak{I}_c\cup\mathfrak{I}_p\cup\mathfrak{J}_{p\& c})\setminus\{\mathcal{C}_c,\mathcal{C}_p\}$ forms a sound and complete axiomatisation for the restriction of the $(\cia\cup\pia,\cia)$-implication problem to disjoint atoms.
\end{theorem}
\begin{proof}
    The soundness of the axiomatisation again follows from the soundness of the inference rules.
    We show that the axiomatisation is complete. Let $\Sigma$ be a set of disjoint CIAs and PIAs, and $X\boto_c Y$ a disjoint CIA. Suppose that $\Sigma\not\vdash X\boto_c Y$. We show that $\Sigma\not\models X\boto_c Y$. 
We may assume that $X\boto_c Y$ is minimal in the sense that for all nonempty $X'\subseteq X$ and $Y'\subseteq Y$ such that $X'Y'\neq XY$, we have $\Sigma\vdash X'\boto_c Y'$. If $X\boto_c Y$ is not minimal, we can remove attributes from $X$ and $Y$ until we obtain a minimal atom. It suffices to show the claim for the minimal atom as decomposition $\mathcal{D}_c$ ensures that the claim must hold also for the original atom. Note that by trivial independence $\mathcal{T}_c$ and symmetry $\mathcal{S}_c$, neither $X$ nor $Y$ is empty.
    
    Let $Z=R\setminus XY$ and $\dom(A)=\{0,1,\nullsymb\}$ for all $A\in R$. Define then $\dom^*(A)=\dom(A)$ for all $A\in XY$ and $\dom^*(A)=\{0,\nullsymb\}$ for all $A\in Z$. Let $A_1\in X$, and define $r_1=\{t\in \prod_{A\in R}(\dom^*(A)\setminus\{\nullsymb\})\mid t(A_1)=\sum_{A\in R\setminus A_1}t(A) \mod 2\}$ and $r_2=\{t(\nullsymb/A_1)\mid t\in r_1\}$. Let then $r=(r_1\cup r_2,1)$. 
Now $r\not\models X\boto_c Y$, because the null values can be filled in such a way that the resulting relation is just $r'=(r_1,2)$ and hence, there are tuples $t_1,t_2\in r'$ such $t_1(A_1)=1$, $t_1(A)=0$ for $A\in X\setminus A_1$, and $t_2(A)=0$ for $A\in Y$, but there is no $t\in r'$ such that $t(X)=t_1(X)$ and $t(Y)=t_2(Y)$ because for all $t\in r'$, $t(A_1)=\sum_{A\in R\setminus A_1}t(A) \mod 2$.
Now we show that $r\models\Sigma$. First note that it is possible to fill the null values in $r$ such that the resulting relation is $(\prod_{A\in R}(\dom^*(A)\setminus\{\nullsymb\}),1)$, and therefore all of the disjoint PIAs in $\Sigma$ hold. 
Note then that any way of filling the null values in $r$ (with 0 and 1) results in a relation such that for any $U\subseteq R$, if $U\subseteq Z$, then every attribute in $U$ is constant, and if $U\cap Z=\emptyset$ and $XY\not\subseteq U$,  then $r(U)=(\prod_{A\in U})(\dom(A)\setminus\{\nullsymb\},m)$ for some $m$. This means that showing that all of the disjoint CIAs in $\Sigma$ hold can be done as in the proof of Theorem \ref{thm-complete} in \cite{geiger:1991}. 

\iffullversion
We include the proof for the sake of comprehensiveness. 
Suppose that $V\boto_c W\in\Sigma$. Assume first that $VW\cap XY=\emptyset$. Then $VW\subseteq Z$, so by the definition of $r$, every attribute in $VW$ is constant. Thus clearly, $r\models V\boto_c W$.

     Assume then that $VW\cap XY\neq\emptyset$ and $XY\not\subseteq VW$. We show that $r\models V\boto_c W$. Because every attribute in $Z$ is constant, it suffices to check that $r\models V\setminus Z\boto_c W\setminus Z$. But since $r((VW)\setminus Z)=(\prod_{A\in (VW)\setminus Z}(\dom(A)\setminus\{\nullsymb\}),m)$ for the function $m$ such that $m(t)=\sum_{s((VW)\setminus Z)=t}1(s)$, clearly $r\models V\setminus Z\boto_c W\setminus Z$.

     Assume finally that $XY\subseteq VW$. We show that this results in a contradiction. Denote $V=X'Y'Z'$ and $W=X''Y''Z''$, where $X=X'X''$, $Y=Y'Y''$, and $Z'Z''\subseteq Z$. By the minimality of $X\boto_c Y$ and symmetry $\mathcal{S}_c$, we have $\Sigma\vdash X'\boto_c Y'$ and $\Sigma\vdash Y\boto_c X''$. From $\Sigma\vdash X'Y'Z'\boto_c X''Y''Z''$, we obtain $\Sigma\vdash X'Y'\boto_c X''Y''$ by using decomposition $\mathcal{D}_c$. Then by applying exchange $\mathcal{E}_c$ to $\Sigma\vdash X'\boto_c Y'$ and $\Sigma\vdash X'Y'\boto_c X''Y''$, we obtain $\Sigma\vdash X'\boto_c X''Y'Y''$. Then by symmetry $\mathcal{S}_c$, we obtain $\Sigma\vdash YX''\boto_c X'$. Then by applying exchange $\mathcal{E}_c$ to $\Sigma\vdash Y\boto_c X''$ and $\Sigma\vdash YX''\boto_c X'$, we obtain $\Sigma\vdash Y\boto_c X''X'$. Then by applying symmetry $\mathcal{S}_c$, we have $\Sigma\vdash X\boto_c Y$.
     \fi
\end{proof}
Moreover, since no rule in $(\mathfrak{I}_c\cup\mathfrak{I}_p\cup\mathfrak{J}_{p\& c})\setminus\{\mathcal{C}_c,\mathcal{C}_p\}$ is such that it has a PIA in the antecedent and a CIA in the consequent, in the disjoint case, the  PIAs in $\Sigma$ do not affect whether a  CIA is logically implied or not, as described in the following theorem.
\begin{theorem}\label{poscertthm}
    Let $\Sigma$ be a set of disjoint CIAs and PIAs, and $\sigma$ a disjoint CIA. Then
    $\Sigma\models\sigma$ if and only if
    $\Sigma\setminus\{\tau\in\Sigma\mid \tau\text{ is a PIA}\}\models\sigma$.
\end{theorem}
The following example demonstrates that the disjointness assumption in the previous two theorems is crucial.  

\begin{example}
Let $\Sigma=\{A\boto_p A,B\boto_p B,C\boto_p C, A\boto_c C,B\boto_c C\}.$
Then we have $\Sigma\not\vdash_{\mathfrak{I}_c\cup\mathfrak{I}_p\cup\mathfrak{J}_{p\& c}} AB\boto_c C$, but $\Sigma\models AB\boto_c C$. For the first part, note that no rule in $\mathfrak{I}_c\cup\mathfrak{I}_p\cup\mathfrak{J}_{p\& c}$ with a PIA in the antecedent has a CIA in the consequent. Therefore it suffices to show that  $\{A\boto_c C,B\boto_c C\}\not\vdash_{\mathfrak{I}_c} AB\boto_c C$. By Theorem \ref{thm-completecert}, this follows from the fact that $\{A\boto_c C,B\boto_c C\}\not\models AB\boto_c C$, which is witnessed by the relation $r'$ of Table \ref{exchangefig}.


To see that \( \Sigma \models AB \boto_c C \), note that each of \( A \), \( B \), and \( C \) is either constant or contains at most one non-null value. Without loss of generality, assume this value is 0. Since the domain of each attribute contains at least two non-null values, any nulls can be filled so that each column has zeroes except at one position, which we set to 1. If there are no nulls, the column consists entirely of zeroes.
If $\Sigma\models C\boto_c C$, we clearly have $\Sigma\models AB\boto_c C$. Now assume, that $\Sigma\not\models C\boto_c C$. Then $\Sigma\vdash A\boto_c C$ and $\Sigma\not\models C\boto_c C$ imply that the column $A$ is all zeroes at the beginning, i.e., $\Sigma\models A\boto_c A$. (If both $\Sigma\not\models A\boto_c B$ and $\Sigma\not\models C\boto_c C$ hold, then the corresponding columns can be filled such that they have zeroes in every position except one position in each column where the value is one, as described above. Then $\Sigma\vdash A\boto_c C$ cannot hold.) Similarly, $\Sigma\vdash B\boto_c C$ and $\Sigma\not\models C\boto_c C$ imply that the column $B$ is all zeroes at the beginning, i.e., $\Sigma\models B\boto_c B$. But this means that $\Sigma\models AB\boto_c AB$, and hence $\Sigma\models AB\boto_c C$.
\end{example}

\begin{corollary}\label{thm-notcomplete}
The set $\mathfrak{I}_c\cup\mathfrak{I}_p\cup\mathfrak{J}_{p\& c}$ does not form a complete axiomatisation for the $(\cia\cup\pia,\cia)$-implication problem.
\end{corollary}


\section{Computational Complexity}\label{sec:complexity}

This section covers the computational complexity of PIAs and CIAs with respect to some key problems. 
We start our analysis with the combined complexity and the data complexity of model checking and conclude by considering the implication problem. 
In this section we assume that the domains of attributes are finite. The results do not depend on how the multiplicities of the tuples are encoded; that is, they can be written either in unary or in binary.



Let $\calP$ be a class of dependencies. The \emph{combined complexity problem for $\calP$} is to decide, given a relation $r$ over attributes $A_1, \dots ,A_n$, the associated (finite) domains $\dom(A_1), \dots ,\dom(A_n)$, and a dependency $\sigma$ from $\calP$ as the input, whether $r$ satisfies $\sigma$. 
If the input dependency $\sigma$ is fixed, the problem is called the \emph{data complexity problem of $\sigma$}. Given a complexity class $\textbf{C}$, we say that the \emph{data complexity problem for $\calP$} is
\begin{itemize}
    \item in $\textbf{C}$ if for any $\sigma \in \calP$, the data complexity problem of $\sigma$ is in $\textbf{C}$;
    \item $\textbf{C}$-hard if for some $\sigma \in \calP$, the data complexity problem of $\sigma$ is $\textbf{C}$-hard; and
    \item $\textbf{C}$-complete when it is both in $\textbf{C}$ and $\textbf{C}$-hard.
\end{itemize}
First we show that for PIAs already data complexity is $\NP$-complete.

\begin{theorem}
        The combined complexity and data complexity problems for possible independence are both \NP-complete. The \NP-hardness holds for any subclass of PIAs that contains an IA of the form $A \boto_p BC$, where $A,B,C$ are distinct attributes.
\end{theorem}
\begin{proof}
        The membership in $\NP$ is straightforward for combined complexity (and therefore also for data complexity). 
        For the \NP-hardness, it suffices to consider only data complexity.  Letting $\sigma \coloneqq VP \boto_p C$,
        we construct a reduction from the satisfiability problem (\SAT) to the data complexity of $\sigma$. The input of $\SAT$ is a Boolean formula in conjunctive normal form: $\phi=C_1 \land \dots \land C_m$, where $C_i= l_{i,1}\lor \dots \lor l_{i,n}$, $i\in [m]$, are such that each $l_{i,j}$, $j\in [n]$, is a propositional variable $p$ or a negated propositional variable $\overline p$. 
    We construct a relation $r$ from $\phi$ in the following way. Assuming that $W$ is  the set of variables appearing in $\phi$, and representing each tuple $t:(V,P,C)\mapsto (a,b,c)$ simply as $(a,b,c)$, we construct the relation $r$ as follows:
    \begin{enumerate}
        \item\label{it:eka1}  For each $p \in W$, add $(\nullsymb,+,p)$ and $(\nullsymb,-,p)$ to $r$. Furthermore, for all $q\in W\setminus \{p\}$, add $(q,\nullsymb,p)$ and $(\overline{q},\nullsymb,p)$ to $r$.
        \item\label{it:toka2} For each clause $C_i = l_{i,1}\lor \dots  \lor l_{i,n}$, add $(\overline{l_{i,1}},\nullsymb,i),\dots, (\overline{l_{i,n}},\nullsymb,i)$ to $r$. Add also $(\nullsymb,\nullsymb,i)$ with multiplicity $n-1$, and $(\nullsymb,+,i)$ with multiplicity $1$ to $r$. Furthermore, assuming $l_{i,1}, \dots , l_{i,n}$ are literals over variables of some set $W'$, add $(v,\nullsymb,i)$ and $(\overline{v},\nullsymb,i)$ for all $v\in W \setminus W'$.
    \end{enumerate}
    Above it is to be understood that the doubly-nested negation $\overline{\overline{p}}$ of a variable $p$ is $p$ itself. Also the multiplicity of each tuple is assumed to be $1$ unless otherwise stated. The construction of $r$ is illustrated in \Cref{tab:reduction}.
        It suffices to show that $\phi$ is satisfiable if and only if $r\models VP\boto_p C$.

    Suppose first $\phi$ is satisfiable by some variable assignment $s$. 
    We construct a grounding $r'$ of the relation $r$ in the following way. We ensure that each tuple has the form $(l,+,*)$ when $s(l)=1$, and $(l,-,*)$ when $s(l)=0$, where $*$ represents an arbitrary value of the variable $C$. It is easy to see that this is possible whenever $s$ satisfies $\phi$. Moreover, this way the obtained grounding $r'$ satisfies the IA $VP\boto C$.

    For the converse direction, suppose $r$ satisfies $VP \boto_p C$, and let $r'$ be the grounding of $r$ satisfying $VP\boto C$. Let us first consider how the tuples introduced in \Cref{it:eka1} are grounded. For each variable $p$ there are only two possibilities: either the pair $(p,+,v)$ and $(\overline{p},-,v)$, or the pair $(p,-,v)$ and $(\overline{p},+,v)$ must appear in the grounding  $r'$ consistently for all values $v$ of $C$. Thus the grounding of \Cref{it:eka1} represents an assignment $s$ of the variables in $\phi$. Before considering \Cref{it:toka2}, let us first note that the values of pairs $(V,P)$ must be consistent between the tuples obtained from \Cref{it:eka1} and those obtained from \Cref{it:toka2}; this ensured by the IA. It is then easy to see that a grounding for the tuples obtained from \Cref{it:toka2} entails $s$ is a satisfying assignment. In particular, for each clause $C_i$, having one tuple of the form $(\nullsymb,+,i)$ in $r$ ensures that for at least one literal $l$ there is a tuple of the form $(l,+,i)$ in the grounding. This concludes the proof.
    \end{proof}
\begin{table}[h!]
\parbox{.45\linewidth}{
\centering
\scalebox{.9}{
\begin{tabular}{ccccc}
\begin{tabular}{cccc}
&\( V \) & $P$ & \( C \) \\\cline{2-4}
&$\nullsymb$&+&$p_1$\\
&$\nullsymb$&-&$p_1$\\
&$p_2$ & $\nullsymb$ & $p_1$\\
&$\overline{p_2}$ & $\nullsymb$ & $p_1$\\
&$p_3$ & $\nullsymb$ & $p_1$\\
&$\overline{p_3}$ & $\nullsymb$ & $p_1$\\\cdashline{2-4}
&$p_1$ & $\nullsymb$ & $p_2$\\
&$\overline{p_1}$ & $\nullsymb$ & $p_2$\\
&$\nullsymb$&+&$p_2$\\
&$\nullsymb$&-&$p_2$\\
&$p_3$ & $\nullsymb$ & $p_2$\\
&$\overline{p_3}$ & $\nullsymb$ & $p_2$\\\cdashline{2-4}
&$p_1$ & $\nullsymb$ & $p_3$\\
&$\overline{p_1}$ & $\nullsymb$ & $p_3$\\
&$p_2$ & $\nullsymb$ & $p_3$\\
&$\overline{p_2}$ & $\nullsymb$ & $p_3$\\
&$\nullsymb$&+&$p_3$\\
&$\nullsymb$&-&$p_3$\\
\end{tabular}
     & 
     $\cup$
\begin{tabular}{cccc}
&\( V \) & $P$ & \( C \) \\\cline{2-4}
&${p_1}$& $\nullsymb$ & $1$\\
&$\overline{p_1}$& $\nullsymb$ & $1$\\
&$\overline{p_2}$& $\nullsymb$ & $1$\\
&$\overline{p_3}$& $\nullsymb$ & $1$\\
&$\nullsymb$& $\nullsymb$ & $1$\\
&$\nullsymb$& $+$ & $1$\\\cdashline{2-4}
&$\overline{p_1}$& $\nullsymb$ & $2$\\
&${p_2}$& $\nullsymb$ & $2$\\
&$\overline{p_3}$& $\nullsymb$ & $2$\\
&$\nullsymb$& $\nullsymb$ & $2$\\
&$\nullsymb$& $\nullsymb$ & $2$\\
&$\nullsymb$& $+$ & $2$\\\cdashline{2-4}
&${p_1}$& $\nullsymb$ & $3$\\
&$\overline{p_1}$& $\nullsymb$ & $3$\\
&${p_2}$& $\nullsymb$ & $3$\\
&$\overline{p_2}$& $\nullsymb$ & $3$\\
&${p_3}$& $\nullsymb$ & $3$\\
&$\nullsymb$& $+$ & $3$\\
\end{tabular}
\end{tabular}
}
\vspace{3mm}
\caption{The relation $r$ obtained via the reduction from an example SAT instance $\phi=C_1 \land C_2\land C_3$, where $C_1= p_2 \lor  p_3$,
$C_2= p_1 \lor \overline{p_2} \lor p_3$, and
$C_3= \overline{p_3}$.
\label{tab:reduction}}
}
\hfill
\parbox{.45\linewidth}{
    \centering
        \[
\begin{array}[t]{c@{\hspace{1cm}}c}
\begin{array}[t]{cccc}
r=&A & B \\\cline{2-3}
&0 & \nullsymb \\
&1 & \nullsymb \\
&\nullsymb & 2 \\
&0 & \nullsymb \\
&\nullsymb & 1 \\
&\nullsymb & \nullsymb \\
&\nullsymb & 0 \\
\end{array}
&
\begin{array}[t]{ccc}
r_\times =&A & B \\\cline{2-3}
&0 & 0 \\
&1 & 0 \\
&1 & 2 \\
&0 & 1 \\
&1 & 1 \\
&0 & 2 \\
\end{array}
\end{array}
\]
\caption{One possibility for the relation $E$ is to connect tuples on the same row from both tables. }
    \label{fig:rtimes}}
\end{table}
Although the data complexity and model checking problems are generally $\NP$-complete, it can be shown that for \emph{unary} PIAs (i.e., PIAs $X \boto_p Y$ where $|X|=|Y|=1$) model checking lies in polynomial time. As shown next, this is achieved by representing the question as a maximum flow problem. 
Let us start by introducing few necessary definitions. A \emph{network} is a directed graph $G = (V, E)$, where each edge $(u, v) \in E$ is associated with a \emph{capacity} $c(u, v) \geq 0$. The network is called a \emph{flow network} if it is furthermore associated with a source node $s\in V$ and a sink node $t \in V$. A \emph{flow} is a mapping $f:E \to \mathbb{R}_{\geq 0}$ that satisfies:
\begin{enumerate}
    \item \emph{capacity constraint}: $f(u, v) \leq c(u, v)$ for all $(u, v) \in E$,
    \item \emph{flow conservation}: $\sum_{(v,u) \in E} f(v, u) = \sum_{(u,v)\in E} f(u, v)$ for all $u \in V \setminus \{s, t\}$.
\end{enumerate}
The \emph{total flow} out of the source $s$ is defined as 
\(
\sum_{(s,v)\in E} f(s, v).
\)
The \emph{maximal flow problem} is to determine the maximum total flow in a given flow network. In what follows we use the well-known integral flow theorem, which states that maximum total flow is integral if all the edge capacities are integral (see, e.g., \cite{0069809}).

\begin{theorem}
    The data complexity and combined complexity problems for unary possible independence are in polynomial time.
\end{theorem}
    \begin{proof}
        Let $r=(r',m)$ be a relation and $A\boto_p B$ a PIA over single attributes $A$ and $B$. 
        We write 
        $r_\times$ for the Cartesian product of the non-null (set) projections of $r$ on $A$ and $B$; that is, $r_\times$ is the set of tuples $t$ over $\{A,B\}$ such that $t(A)\in r'(A)\setminus \{\nullsymb\}$ and $t(B) \in r'(B)\setminus \{\nullsymb\}$. 
         The flow network is then constructed in the following way. The nodes consist of the tuples in $r'$ and $r_\times$, as well as fresh values $s_0$ and $s_1$ denoting the source and the sink, respectively. The edges and capacity constraints are build as follows:
        \begin{enumerate}
            \item\label{item:yks} For each $t \in r'$ there is an edge $(s_0,t)$ with capacity constraint $f(s_0,t) \leq m(t)$.
            \item\label{item:kaks} For each $t\in r'$ and $t'\in r_\times$ such that $t'$ is a grounding of $t$, there is an edge $(t,t')$ with capacity constraint $f(t,t')\leq 1$
            \item\label{item:kolme} For each $t \in r_\times$ there is an edge $(t,s_1)$ with capacity constraint $f(t,s_1)\leq 1$.
        \end{enumerate}
        We claim that the maximum total flow is $|r_\times|$ if and only if $r \models A\boto_p B$. 
Suppose first the maximum total flow $\sum_{v \in V} f(s, v)$ is $|r_\times|$. Consequently, $f$ assigns each incoming edge $(t,s_1)$ of the sink the maximum capacity: $f(t,s_1)=1$. Due to flow conservation and the the integral flow theorem, for each tuple $t'\in r_\times$ we find exactly one tuple $t\in r'$ such that $f(t,t')=1$. Conversely, due to \Cref{item:yks} and flow conservation, a tuple $t\in r'$ is associated with no more than $m(t)$ many tuples $t'\in r_\times$ such that $f(t,t')=1$.  Since $f(t,t')=1$  implies $t'$ is a grounding of $t$, we obtain $r \models A \boto B$.
        For the converse direction, assume $r \models A \boto B$. Then we find a grounding $r'$ of $r$ satisfying $A \boto B$.
        Consequently, the set $r_\times$, defined from $r$ as above, is a subset of $r'$. It suffices to define a flow function $f$ which has a total flow $|r_\times|$. Note that the maximum total flow cannot be greater than $|r_\times|$ due to the capacity constraints in \Cref{item:kolme}.
        
        Now, it is possible to choose a binary relation $E \subseteq r'\times r_\times$ (see \Cref{fig:rtimes}) such that             (1) for each $t\in r'$ there is at most $m(t)$ distinct $t'\in r_\times$ such that $(t,t')\in E$; and (2) for each $t\in r$ there is exactly one $t'\in r'$ such that $(t',t)\in E$.
        We let $f(t,t')=1$ if $(t,t')\in E$, and $f(t,t')=0$ if $(t,t')\in r'\times r_\times \setminus E$. The remaining values of $f$, for the edges outgoing the source and the edges incoming the sink, are uniquely determined by flow conservation. It follows by construction $f$ satisfies all the capacity constraints. We conclude that the total flow of $f$ is $|r_\times|$, as required.
     \end{proof}   
We have now shown that the data complexity of model checking is $\NP$-complete for the class of PIAs $X \boto Y$ where $X$ or $Y$ contains at least two attributes. When both $X$ and $Y$ are single attributes, we established that the problem is in polynomial time. We conclude this section by showing that for CIAs data complexity is in \FO. For the next proposition, recall that the domain of each attribute has at least two elements.


\begin{proposition}\label{prop:easy}
    $r \models X \boto_c X$ if and only if $r(X)$ is complete and satisfies $X\boto X$.
\end{proposition}
    
\begin{theorem}
   The data complexity problem for IAs and CIAs is in \FO. The combined complexity problem for CIAs is in polynomial time.
\end{theorem}
\begin{proof}
Clearly it suffices to show the claim for CIAs.
We claim that $r \models X \boto_c Y$ iff one of the following three properties hold: 

(i) $r \models X \boto_c X$, \qquad
    (ii) $r \models Y \boto_c Y$,  \qquad
   (iii) both of the following hold:
    \begin{enumerate}
    \item\label{it:eka} for all complete $t\in r(X)$ and $t'\in r(Y)$, there is $t''\in r$ such that $t''(X)=t$ and $t''(Y)=t'$; and
    \item\label{it:toka} any grounding of $r(XY)$ is included in $r(XY)$.
\end{enumerate}
We show first that, assuming $r \not\models X \boto_c X$ and $r \not\models Y \boto_c Y$, the failure of either \Cref{it:eka} or \Cref{it:toka} leads to $r\not\models X \boto_c Y$. Suppose \Cref{it:eka} is not true.
Let $t\in r(X)$ and $t'\in r(Y)$ be two complete tuples such that $t''(X)\neq t$ or $t''(Y)\neq t'$ for any $t''\in r$. Clearly, for each $t''\in r$ we can then define a grounding $s$
 such that $s(X)\neq t$ or $s(Y)\neq t'$. For this, note that the domain of each attribute is at least of size $2$. 
 The obtained grounding of $r$ then does not satisfy $X \boto Y$, which means that $r$ does not satisfy $X \boto_c Y$.

Suppose then \Cref{it:eka} holds but \Cref{it:toka} does not hold. Let $s\notin r(XY)$ be a grounding of a tuple $t\in r(XY)$. \Cref{it:eka} entails that $s(X)\notin r(X)$ or $s(Y)\notin r(Y)$.
By symmetry, we may assume that $s(X)\notin r(X)$. Since $r \not\models Y \boto_c Y$, either $r \not\models Y\cap X \boto_c Y\cap X$ or $r \not\models Y\setminus X \boto_c Y\setminus X$. Since $r \not\models Y\cap X \boto_c Y\cap X$ entails $r \not\models X \boto_c Y$ by the decomposition and symmetry rules of certain independence, we may assume that $r \not\models Y\setminus X \boto_c Y\setminus X$. We then construct a grounding $r'$ of $r$ with the following properties: 
\begin{itemize}
    \item\label{it:1} One occurrence of the tuple $t$ in $r$ is grounded to $s'$ such that $s'(X)=s(X)$.
    \item\label{it:2} The remaining tuples in $r$ are grounded to tuples $s''$ such that $s''(X)\neq s(X)$.
    \item\label{it:3} The grounding with respect to $Y\setminus X$ is such that $r'\not\models Y\setminus X \boto Y\setminus X$.
\end{itemize}
It is not difficult to see that $r' \not\models X \boto Y$, as required.
For the converse direction, suppose one of the three listed properties hold. In the case of $r \models X \boto_c X$ or $r \models Y \boto_c Y$, we obtain $r \models X \boto_c Y$ by the constancy, symmetry, and trivial independence rules of certain independence. Furthermore, $r \models X \boto_c Y$ is a straightforward consequence of \Cref{it:eka,it:toka}. This concludes the proof of the claim.

A relation $r$ over an ordered list of $n$ attributes can be interpreted as a first-order structure $\frak{A}_r=(U,f, D_1, \dots ,D_n,0)$, where $U\coloneqq D_1\cup \dots \cup D_n \cup \{0, \dots ,m\}$,
$f :U^n \to \{0, \dots, m\}$ is a function representing $r$, and $D_i$, for $i\in [n]$, is the domain of the $i$th attribute in $r$. The attributes of $XY$ furthermore are assumed to occupy fixed positions in the ordered list of attributes. Then, using \Cref{prop:easy} and the claim proven above, it is straightforward to write a first-order formula $\phi$ such that $\frak{A}_r\models \phi$ iff $r \models X \boto_c Y$. We conclude that the data complexity problem for CIAs is in \FO. Regarding combined complexity, the items of the claim can be checked in polynomial time in the combined size of the relation and the CIA.
\end{proof}
From the results of Section \ref{axiomatisations}, we immediately obtain the following theorem concerning the complexity of implication problems, where the problem is to decide, given a finite set $\Sigma\cup\{\tau\}$ of atoms, whether $\Sigma\models\tau$.

\begin{theorem}\label{impl_comp}
The implication problems of Theorems \ref{thm-completecert}, \ref{thm-completepos}, and \ref{thm-completedisjointcertpos_cert} are in polynomial time.
\end{theorem}
    
    

\section{Conclusion and Future Work}\label{sec:conclusion}

We have initiated work on the general concept of independence in the presence of incomplete information. We have shown that results on the implications problem well-known from the idealised special case of complete relations are captured by the concept of certain independence in the general realistic case of incomplete relations. In addition, we have shown that the concept of possible independence is challenging in the sense of establishing general results on their implication problem, and intractable in terms of their model checking problem. 
Directions for future work include the axiomatisability and computational complexity for more general fragments of possible (and certain) independence, the discovery of PIAs and CIAs  from given incomplete relations, and the existence/construction of Armstrong relations. More generally, it is interesting to look at alternative approaches to capturing incomplete information and probabilistic variants of PIAs and CIAs.

\begin{credits}
\subsubsection{\ackname} 
M. Hannula: Partially supported by the European Research Council (ERC) under the European Union’s Horizon 2020 research and innovation programme (grant agreement No 101020762).
M. Hirvonen: Received funding from the European Research Council (ERC) under the European Union’s Horizon 2020 research and innovation programme (grant agreement No 101020762) and from the Magnus Ehrnrooth foundation.
\subsubsection{\discintname}
The authors have no competing interests to declare that are
relevant to the content of this article. 
\end{credits}

\bibliographystyle{splncs04}
\bibliography{biblio}

\newpage
\setcounter{page}{1}
\appendix

\section{Additional Proofs}\label{appendix}

\newtheorem{innercustomthm}{Theorem}
\newenvironment{customthm}[1]
  {\renewcommand\theinnercustomthm{#1}\innercustomthm}
  {\endinnercustomthm}

\begin{customthm}{\ref{thm-completecert}.}
The set $\mathfrak{I}_c$ forms a sound and complete axiomatisation for the implication problem for CIAs.
\end{customthm}
\begin{proof}
    The soundness of the axiomatisation is clear, as it is easy to check that all of the inference rules in $\mathfrak{I}_c$ are sound.
    We show that the axiomatisation is also complete. Let $\Sigma\cup\{\sigma\}$ be a set of \cias. Suppose that $\Sigma\not\vdash_{\mathfrak{I}_c} \sigma$. We show that $\Sigma\not\models \sigma$. 
    It is clear that $\Sigma\vdash_{\mathfrak{I}_c}\sigma$ if and only if $ind(\Sigma)\vdash_{\mathfrak{I}}ind(\sigma)$, because rules of $\mathfrak{I}_c$ and $\mathfrak{I}$ are the same, except that in $\mathfrak{I}_c$, we have CIAs instead of IAs. Thus, $ind(\Sigma)\not\vdash_{\mathfrak{I}}ind(\sigma)$, and by the completeness of the axiomatisation $\mathfrak{I}$, we have $ind(\Sigma)\not\models ind(\sigma)$. This means that there is a relation $r$ over the relation schema $R=\{A\in\att\mid A \text{ appears in some atom of } \Sigma\cup\{\sigma\}\}$ such that $r\models ind(\Sigma)$, but $r\not\models  ind(\sigma)$. Moreover, since $ind(\Sigma\cup \{\sigma\})$ is a set of IAs, the relation $r$ must be complete. This means that there are no null symbols in $r$, and therefore $r\models ind(\Sigma)$ implies $r\models \Sigma$, and $r\not\models  ind(\sigma)$ implies $r\not\models \sigma$. Hence, the relation $r$ witnesses that $\Sigma\not\models \sigma$, and the axiomatisation $\mathfrak{I}_c$ is complete.
\end{proof}

\begin{customthm}{\ref{equiv}.}
    Let $\Sigma\cup\{\sigma\}$ be a set of CIAs.  Then
    $\Sigma\models\sigma \iff ind(\Sigma)\models ind(\sigma).$
\end{customthm}
\begin{proof}
    As stated in the proof of Theorem \ref{thm-completecert}, it is clear that $\Sigma\vdash_{\mathfrak{I}_c}\sigma$ if and only if $ind(\Sigma)\vdash_{\mathfrak{I}}ind(\sigma)$ 
    The claim then follows from Theorems \ref{thm-complete} and \ref{thm-completecert}.
\end{proof}

\begin{customthm}{\ref{armstr}.}
    The class of CIAs enjoys Armstrong relations.
\end{customthm}
\begin{proof}
    Let $R$ be a relation schema and $\Sigma$ a set of CIAs over $R$. Then there exists a complete relation $r$ such that for all CIAs $\sigma$ over $R$, we have
    \[
    \begin{array}{l}
   \hspace{-.2cm} \Sigma\models\sigma  \iff\sigma\in\cl_{\mathfrak{I}_c}(\Sigma)
    \iff ind(\sigma)\in\cl_{\mathfrak{I}}(ind(\Sigma))\\
    \hspace{-.2cm} \iff ind(\Sigma)\models ind(\sigma)
    \iff r\models ind(\sigma)
    \iff r\models \sigma.
    \end{array}
    \]
    The first equivalence follows from Theorem \ref{thm-completecert}, and the second one from the fact that $\Sigma\vdash_{\mathfrak{I}_c}\sigma$ iff $ind(\Sigma)\vdash_{\mathfrak{I}}ind(\sigma)$. The third equivalence follows from Theorem \ref{thm-complete}, and the fourth one from the definition of Armstrong relations for IAs. Note that the Armstrong relation $r$ for $ind(\Sigma)$ exists by Theorem \ref{thm-armstrong-rel}. The last equivalence follows from the completeness of $r$.
\end{proof}
\begin{customthm}{\ref{thm-completepos}.}
The set $\mathfrak{I}_p$ 
forms a sound and complete axiomatisation for the $(\pia,\pia^*)$-implication problem, where $\pia^*$ is the class of \pias $X\boto_p Y$ such that at least one of the following conditions hold:

\centering
    (i) $|X|=1$ or $|Y|=1$, \qquad or \qquad (ii) $||X|-|Y||\leq 1$.
\end{customthm}
\begin{proof}
    The soundness of the axiomatisation is again easy to show by checking that all of the inference rules are sound.
    We show that the axiomatisation is complete. Let $\Sigma$ be a set of PIAs and $X\boto_p Y$ a PIA. Suppose that $\Sigma\not\vdash X\boto_p Y$. We show that $\Sigma\not\models X\boto_p Y$. 

    We first consider the case of the constancy atom, so assume that $X=Y=B$. Let $\dom(A)=\{0,1,\nullsymb\}$ for all $A\in R$, and define $\dom^*(B)=\dom(B)=\{0,1,\nullsymb\}$ and $\dom^*(A')=\{0,\nullsymb\}$ for all $A'\in R\setminus B$. Define then $r'=\prod_{A\in R}(\dom^*(A)\setminus\{\nullsymb\})$, and $r=(r',1)$. Now clearly $r\not\models B\boto_p B$. We show that $r\models \Sigma$. Let $V\boto_p W\in\Sigma$. Now $B\not\in V\cap W$, because otherwise $\Sigma\vdash B\boto_p B$ by decomposition $\mathcal{D}_p$. This means that either $V\subseteq R\setminus B$ or $W\subseteq R\setminus B$, i.e., either all the columns in $V$ or all the columns in $W$ are constant in $r$. Then $r\models V\boto_p W$ by the soundness of trivial independence $\mathcal{T}_p$, symmetry $\mathcal{S}_p$ and constancy $\mathcal{C}_p$.
    
   We may then assume that $X\cap Y\neq \emptyset$. Otherwise for any $A\in X\cap Y$ either $\Sigma\not\vdash A\boto_p A$ or $\Sigma\vdash A\boto_p A$. In the first case, the construction from the case where $X=Y=B$ witnesses that $\Sigma\not\models A\boto_c A$. Hence, (by the decomposition $\mathcal{D}_p$) also $\Sigma\not\models X\boto_p Y$. In the second case, $\Sigma\vdash A\boto_c A$, we can use constancy $\mathcal{C}_p$ to obtain $\Sigma\not\vdash X\setminus\{A\}\boto_p Y\setminus\{A\}$, and it suffices to show the claim for this atom. Note that by the latter argument, we may also assume that there are no $A\in XY$ such that $\Sigma\vdash A\boto_p A$.

The case that $X\boto_p Y$ is disjoint can be handled in two parts. The first case that $|X|-|Y|\leq 1$ and $|Y|\geq 2$. The second case is that $|Y|=1$. Note that by symmetry we may assume that $|X|\geq |Y|$, so it suffices to consider these two cases.

    For the first case, let $|X|=k\geq m=|Y|$, and assume that $m\geq 2$. Let $\dom(A)=\{0,1,\nullsymb\}$ for all $A\in R$. Construct then a relation $r$ over $R$ as follows. The relation $r$ consists of $2^{k}(2^m-1)-1$ rows. On the first $2^{k}$ rows, let the values of $X$ be the tuples from $\{0,1\}^{k}$. For the first $2^m-1$ rows of these $2^{k}$ rows, let the values of $Y$ be the tuples from $\{0,1\}^m\setminus \{(1,1,\dots,1)\}$ and for the rest of the $2^{k}-(2^m-1)$ rows, let the values of $Y$ be null symbols. Then add rows with only null symbols such that you obtain $2^{k}(2^m-1)-1$ rows in total. Let all the values for the attributes $A\in R\setminus XY$ be constant 0.
    
    We show that $r\models\Sigma$, but $r\not\models X\boto_p Y$. First note that $r(X)$ has $2^{k}$ different complete tuples, and $r(Y)$ has $2^m-1$ different complete tuples. 
    Therefore, it is impossible for $X\boto_p Y$ to hold in the relation $r$, because $r$ has only $2^{k}(2^m-1)-1$ tuples.
    
    Suppose then that $V\boto_p W \in\Sigma$. We will show that $r\models V\boto_p W$. Since every attribute $A\in R\setminus XY$ is constant, we may assume that $VW\subseteq XY$, and therefore also $V\cap W=\emptyset$. (For the latter, recall that there are no $A\in XY$ such that $\Sigma\vdash A\boto_p A$.) Moreover, we may assume that $VW=XY$, because by decomposition $\mathcal{D}_p$, $r\models V\boto_p W$ implies $r\models V'\boto_p W'$ for every $V'\subseteq V$ and $W'\subseteq W$. 
    
    We may assume that $|V|\geq|W|$. Suppose first that $|V|\geq k+1$. This means that $r(V)$ has at most $2^{k}$ different non-null tuples and $r(W)$ has at most $2^{k+m-|V|}\leq 2^{m-1}$ different non-null tuples. Then in order to $V\boto_p W$ hold, we need to fit at most $2^{k}\cdot 2^{m-1}$ values of $XY$ in the relation. Since we have $2^{k}(2^m-1)-1(>2^{k}\cdot 2^{m-1})$ rows in $r$, and no non-null tuple is repeated in $r$, we have enough room to contain all the needed rows. This can be done by a grounding $r'$ of $r(XY)$ that replaces all the null symbols on the first $2^k$ rows with 0s, and replaces the null symbol rows with tuples from the Cartesian product of the projections of these first $2^k$ rows on $V$ and $W$. Since replacing all the null symbols on the first $2^k$ rows with 0s preserves the amounts of different non-null tuples, the cardinality approximations above hold for $r'$. Thus $r\models V\boto_p W$.

    Suppose then that $|V|=k$. Since $|V|+|W|=|VW|=|XY|=|X|+|Y|=k+m$, we have $|W|=m$. Since $V\neq X$ and $W\neq X$, there are $A\in V$ and $B\in W$ such that $A,B\in Y$. As the relation $r$ does not contain the tuple where every attribute of $Y$ has value 1, it must be that by picking suitable values for the null symbols on the first $2^k$ row, we obtain at most $2^{k}-1$ different values for $V$ and $r(W)$ has at most $2^m-1$ different values for $W$. Since no non-null tuple is repeated in $r$, we have enough room to contain all the needed rows. This can be done by a grounding $r'$ of $r(XY)$ that replaces the null symbols on the first $2^k$ rows such that we obtain at most $2^{k}-1$ and $2^{m}-1$ different values for $V$ and $W$, respectively. The grounding $r'$ replaces the null symbol rows with tuples from the Cartesian product of the projections of these first $2^k$ rows on $V$ and $W$, as before. Thus $r\models V\boto_p W$.

    Note that it cannot be that $|V|\leq k-1$. Otherwise, from $|X|-|Y|\leq 1$ and $|X|=k$, it follows that  $k+(k-1)\leq|X|+|Y|=|XY|=|VW|=|V|+|W|\leq k-1+|W|$, i.e., $|W|\geq k$, which is impossible because $|W|\leq|V|\leq k-1$. This finishes the proof in the case that $m\geq 2$.

    We now consider the second case where $|X|= k\geq 1=|Y|$. Let again $\dom(A)=\{0,1,\nullsymb\}$ for all $A\in R$. Construct then $r'$ as follows. The relation $r'$ consists of $2^{k+1}-1$ rows. On the first $2^k$ rows, let the values of $X$ be the tuples from $\{0,1\}^k$. For the first two rows of these $2^k$ rows, let the values of $Y$ be 0 (first row) and 1 (second row). For the rest of the $2^k-2$ rows, let the values of $Y$ be null symbols. Then add rows with only null symbols such that you obtain $2^{k+1}-1$ rows in total. Let all the values for the attributes $A\in R\setminus XY$ be constant 0.

    Now $r\not\models X\boto_p Y$ by an argument similar to the case where $|X|=k\geq m=|Y|$ and $m\geq 2$. Suppose then that $V\boto_p W \in\Sigma$. We will show that $r\models V\boto_p W$. As before, we may assume that $VW=XY$. The proof in the case $|V|\geq k+1$ now follows from the trivial independence $\mathcal{T}_p$, because $|V|\geq k+1$ implies that $W=\emptyset$.

    Suppose then that $|V|=k$ and $|W|=1$. If $k=1$, then $V\boto_p W$ is either $X\boto_p Y$ or $Y\boto_p X$. This is impossible, because $\Sigma\not\vdash X\boto_p Y$. Hence, we may assume that $k>1$. 

    Suppose then that $Y\subseteq V$, and $|V|=l+1>1$ and $|W|=j>0$, where $l+j=k$. Since $Y\subseteq V$ and the only possible non-null values for $Y$ are 0 and  1, by picking the value 0 for the null symbols of the column $Y$ on the first $2^k$ rows, we obtain $2^l+1$ values for $V$. We have at most $2^j$ possible non-null values for $W$. Because $l+j=k$, $j\leq k-1$, and $k> 1$, we have $(2^l+1)2^j=2^k+2^j\leq2^k+2^{k-1}=3\cdot2^{k-1}< 4\cdot2^{k-1}-1=2^{k+1}-1$. Since no non-null tuple is repeated in $r$, we have enough room to contain all the needed rows. This be done by a grounding $r'$ of $r(XY)$ that replaces the null symbols of the column $Y$ on the first $2^k$ rows with 0s, and replaces the null symbol rows with tuples from the Cartesian product of the projections of these first $2^k$ rows on $V$ and $W$, as before.
\end{proof}
\begin{figure}
    \centering
    \begin{subfigure}
       \centering
           \begin{tabular}{c c c c c c c c|c}
       $X_1$ & $X_2$ & $X_3$ & $Y_1$ & $Y_2$ & $Z_1$ & $\dots$ & $Z_n$ & $\#$   \\
       \hline
       0 & 0 & 0 & 0 & 0 & 0 & $\dots$ & 0 & 1   \\
       0 & 0 & 1 & 0 & 1 & 0 & $\dots$ & 0 & 1   \\
       0 & 1 & 0 & 1 & 0 & 0 & $\dots$ & 0 & 1   \\
       0 & 1 & 1 & $\nullsymb$ & $\nullsymb$ & 0 & $\dots$ & 0 & 1   \\
       1 & 0 & 0 & $\nullsymb$ & $\nullsymb$ & 0 & $\dots$ & 0 & 1   \\
       1 & 0 & 1 & $\nullsymb$ & $\nullsymb$ & 0 & $\dots$ & 0 & 1   \\
       1 & 1 & 0 & $\nullsymb$ & $\nullsymb$ & 0 & $\dots$ & 0 & 1   \\
       1 & 1 & 1 & $\nullsymb$ & $\nullsymb$ & 0 & $\dots$ & 0 & 1   \\
       $\nullsymb$ & $\nullsymb$ & $\nullsymb$ & $\nullsymb$ & $\nullsymb$ & 0 & $\dots$ & 0 & 15  \\
    \end{tabular}
 \qquad
\centering
       \begin{tabular}{ c c c c c c | c }
 $X_1$ & $X_2$ & $Y_1$ & $Z_1$ & $\dots$ & $Z_n$ &  $\#$   \\  
 \hline
 0 & 0 & 0 & 0 & $\dots$ & 0 & 1  \\
 0 & 1 & 1 & 0 & $\dots$ & 0 & 1 \\
 1 & 0 & $\nullsymb$ & 0 & $\dots$ & 0 & 1\\
 1 & 1 &  $\nullsymb$  & 0 & $\dots$ & 0 & 1\\
 $\nullsymb$ & $\nullsymb$ & $\nullsymb$ & 0 & $\dots$ & 0 & 3\\
 \end{tabular}
    \end{subfigure}
    \caption{
    The relation $r$ from the proof of Theorem \ref{thm-completepos} in the case $k=3$, $m=2$, where $X=\{X_1,X_2,X_3\}$, $Y=\{Y_1,Y_2\}=$ and $R\setminus XY=\{Z_1,\dots,Z_n\}$, and in the case $k=2$, $m=1$, where $X=\{X_1,X_2\}$, $Y=\{Y_1\}=$ and $R\setminus XY=\{Z_1,\dots,Z_n\}$.
}
    \label{fig:possibleind_app}
\end{figure}
\begin{customthm}{\ref{poscertthm}.}
    Let $\Sigma$ be a set of disjoint CIAs and PIAs, and $\sigma$ a disjoint CIA. Then
    $\Sigma\models\sigma$ if and only if
    $\Sigma\setminus\{\tau\in\Sigma\mid \tau\text{ is a PIA}\}\models\sigma$.
\end{customthm}

\begin{proof}
    Let $\Sigma$ and $\sigma$ be as in the above theorem. By Theorem \ref{thm-completedisjointcertpos_cert}, $\Sigma\models\sigma$ if and only if $\Sigma\vdash_{(\mathfrak{I}_c\cup\mathfrak{I}_p\cup\mathfrak{J}_{p\& c})\setminus\{\mathcal{C}_c,\mathcal{C}_p\}}\sigma$. Since no rule in $(\mathfrak{I}_c\cup\mathfrak{I}_p\cup\mathfrak{J}_{p\& c})\setminus\{\mathcal{C}_c,\mathcal{C}_p\}$ is such that it has a PIA in the antecedent and a CIA in the consequent, we have $\Sigma\vdash_{(\mathfrak{I}_c\cup\mathfrak{I}_p\cup\mathfrak{J}_{p\& c})\setminus\{\mathcal{C}_c,\mathcal{C}_p\}}\sigma$ if and only if $\Sigma\vdash_{\mathfrak{I}_c\setminus\{\mathcal{C}_c\}}\sigma$. The latter is clearly equivalent with $\Sigma\setminus\{\tau\in\Sigma\mid \tau\text{ is a PIA}\}\vdash_{\mathfrak{I}_c\setminus\{\mathcal{C}_c\}}\sigma$. Since all the atoms considered are disjoint, it does not matter whether we consider the set $\mathfrak{I}_c\setminus\{\mathcal{C}_c\}$ or $\mathfrak{I}_c$. Then by Theorem \ref{thm-completecert}, $\Sigma\setminus\{\tau\in\Sigma\mid \tau\text{ is a PIA}\}\vdash_{\mathfrak{I}_c}\sigma$ if and only if $\Sigma\setminus\{\tau\in\Sigma\mid \tau\text{ is a PIA}\}\models\sigma$.
\end{proof}
\begin{customthm}{\ref{impl_comp}.}
The implication problems of Theorems \ref{thm-completecert}, \ref{thm-completepos}, and \ref{thm-completedisjointcertpos_cert} are in polynomial time.
\end{customthm}
\begin{proof}
    (Theorem \ref{thm-completecert}.) By Theorem \ref{equiv}, the implication problem for CIAs is equivalent to the implication problem for IAs, so the cubic time algorithm~\cite{geiger:1991} can be used.
    
(Theorem \ref{thm-completepos}.) The theorem states that the set $\mathfrak{I}_p$ is sound and complete for the $(\pia,\pia^*)$-implication problem. 
    Deciding whether \( X \boto_p Y \) can be derived from \( \Sigma \) using the inference system \( \mathfrak{I}_p \) is in polynomial time. When both \( X \) and \( Y \) are non-empty, it suffices to remove from each the attributes that \( \Sigma \) designates as constants, yielding reduced sets \( X^* \) and \( Y^* \). One then checks whether \( \Sigma \) contains a PIA \( X' \boto_p Y' \) or \( Y' \boto_p X' \) such that \( X^* \subseteq X' \) and \( Y^* \subseteq Y' \). If either \( X \) or \( Y \) is empty, the trivial independence axiom applies.
    
(Theorem \ref{thm-completedisjointcertpos_cert}.)
 Let $\Sigma$ be a set of disjoint CIAs and PIAs, and  let $\sigma$ be a disjoint CIA. By Theorems \ref{equiv} and \ref{poscertthm}, $\Sigma\models\sigma$ if and only if $ind(\Sigma\setminus\{\tau\in\Sigma\mid \tau\text{ is a PIA}\})\models ind(\sigma)$.
    This is an instance of implication problem for IAs, and therefore the problem is in cubic time~\cite{geiger:1991}.
\end{proof}
\end{document}